\documentclass[runningheads,a4paper]{llncs}
\pdfoutput=1 %
\usepackage{orcidlink}\renewcommand{\orcidID}[1]{\orcidlink{#1}}
\usepackage[utf8]{inputenc}
\usepackage[T1]{fontenc}
\usepackage{llncsstuffthm}
\usepackage{cwdefs}
\usepackage{pgfkeys}
\usepackage{plabels}
\usepackage{amssymb}
\usepackage{amsfonts}
\usepackage{multirow}
\usepackage{longtable}
\usepackage{booktabs}
\usepackage{amsmath}
\usepackage{bigdelim}
\usepackage{comment}
\usepackage{xspace}
\usepackage{booktabs}
\usepackage{graphicx}
\usepackage{setspace}
\usepackage{upgreek}
\usepackage{enumitem}
\usepackage{rotating}
\usepackage{array}
\usepackage[titles]{tocloft}
\usepackage{xcolor}
\usepackage{colortbl}
\usepackage{capt-of}
\usepackage{pdfpages}
\usepackage{wrapfig}
\usepackage{fancyvrb}
\usepackage{hyperref}
\usepackage[capitalise]{cleveref}
\usepackage[noadjust,nobreak]{cite}
\hypersetup{
hyperfootnotes=false,
breaklinks=true,
hidelinks,
linktoc=all,
colorlinks=false
}

\crefname{defn}{Definition}{Definitions}
\crefname{proc}{Procedure}{Procedures}
\crefname{thm}{Theorem}{Theorems}
\crefname{examp}{Example}{Examples}
\crefname{subex}{Example}{Examples}
\crefname{coro}{Corollary}{Corrolaries}
\crefname{subthm}{Theorem}{Theorems}

\crefname{section}{Sect.}{Sects.}
\crefname{prop}{Prop.}{Props.}
\crefname{subprop}{Prop.}{Props.}
\Crefname{section}{Section}{Sections}
\Crefname{prop}{Proposition}{Propositions}
\Crefname{subprop}{Proposition}{Propositions}

\crefname{stp}{step}{steps}
\crefname{clm}{claim}{claims}
\crefname{rule}{rule}{rules}

\DeclareRobustCommand{\abbrevcrefs}{%
  \crefname{prop}{Prop.}{Props.}%
  \crefname{subprop}{Prop.}{Props.}%
}
\DeclareRobustCommand{\abbref}[1]{{\abbrevcrefs\cref{#1}}}
\newcommand{\Vampire}{\textsf{Vampire}\xspace}

\newcommand{\EProver}{\textsf{E}\xspace}

\newcommand{\ProverN}{\textsf{Prover9}\xspace}
\newcommand{\leanCoP}{\textsf{leanCoP}\xspace}

\newcommand{\CMProver}{\textsf{CMProver}\xspace}
\newcommand{\SETHEO}{\textsf{SETHEO}\xspace}
\newcommand{\PTTP}{\textsf{PTTP}\xspace}
\newcommand{\PIE}{\textsf{PIE}\xspace}

\definecolor{colorjan}{rgb}{0.8,0.0,0.0}
\definecolor{colorch}{rgb}{0.0,0.0,0.8}

\newcommand{\prenum}{1.0em}

\newcommand{\por}{;}
\newcommand{\pnot}{\mathbf{not}\;}

\newcommand{\rfolzero}{\gamma^0}
\newcommand{\rfolone}{\gamma^1}
\newcommand{\pfol}{\gamma}

\newcommand{\SP}{\mathsf{S}}

\newcommand{\pred}{\mathcal{P}\hspace{-2pt}\mathit{red}}

\newcommand{\predplain}{\mathcal{P}\hspace{-2pt}\mathit{red}^\mathit{LP}}
\newcommand{\predpol}{\mathcal{P}\hspace{-2pt}\mathit{red}^{\pm}}

\newcommand{\fun}{\mathcal{F}\hspace{-2pt}\mathit{un}}

\newcommand{\ren}{\f{rename}_{0\mapsto 1}}

\newcommand{\superscripted}{$0/1$-super\-scripted\xspace}

\newcommand{\algoinput}{\smallskip \noindent\textsc{Input: }}
\newcommand{\algooutput}{\smallskip \noindent\textsc{Output: }}
\newcommand{\algomethod}{\smallskip \noindent\textsc{Method: }}

\newcommand{\VP}{V_{+}}
\newcommand{\VPP}{V_{+1}}
\newcommand{\VN}{V_{-}}

\newcommand{\fn}{\f{n}}
\newcommand{\fz}{\f{z}}
\newcommand{\fu}{\f{u}}

\newcommand{\prooftab}{\small}

\begin{document}

\title{
Synthesizing Strongly Equivalent Logic Programs: Beth Definability for Answer
Set Programs via Craig Interpolation in First-Order Logic%
\thanks{Funded by the Deutsche Forschungsgemeinschaft (DFG, German
  Research Foundation) -- Project-ID~457292495.}
}
\titlerunning{Synthesizing Strongly Equivalent Logic Programs}
\author{Jan Heuer\orcidID{0009-0004-1545-9356} \and
  Christoph Wernhard\orcidID{0000-0002-0438-8829}}
\authorrunning{J. Heuer and C. Wernhard}
\institute{University of Potsdam, Potsdam, Germany\\
  \email{jan.heuer,christoph.wernhard@uni-potsdam.de}}

\maketitle

\begin{abstract}
We show a projective Beth definability theorem for logic programs under the
stable model semantics: For given programs $P$ and $Q$ and vocabulary~$V$ (set
of predicates) the existence of a program~$R$ in $V$ such that $P \cup R$ and
$P \cup Q$ are strongly equivalent can be expressed as a first-order
entailment.
Moreover, our result is effective:
A program~$R$ can be constructed from a Craig interpolant for this entailment,
using a known first-order encoding for testing strong equivalence, which we
apply in reverse to extract programs from formulas.
As a further perspective, this allows transforming logic programs via
transforming their first-order encodings.
In a prototypical implementation, the Craig interpolation is performed by
first-order provers based on clausal tableaux or resolution calculi.
Our work shows how definability and interpolation, which underlie modern
logic-based approaches to advanced tasks in knowledge representation, transfer
to answer set programming.

\end{abstract}

\section{Introduction}

Answer set programming~%
\cite{niemelae:asp:99,marek:truzcynski:1999,handbook:asp:2008,baral:book,lifschitz:book}
is one of the major paradigms in knowledge representation. A problem is
expressed declaratively as a logic program, a set of rules in the form of
implications. An answer set solver returns representations of its
\name{answer sets} or \name{stable models}~%
\cite{gelfond:lifschitz:1988,lifschitz:10:13defs}. That is, minimal Herbrand
models, where models with facts not properly justified in a
non-circular way are excluded. Modern answer set solvers such as \name{clingo}~%
\cite{clingo:2019} are advanced tools that integrate SAT technology.

Two logic programs can be considered as \name{equivalent} if and only if they
have the same answer sets. However, if two equivalent programs are each
combined with some other program, the results are not necessarily equivalent.
Thus, it is of much more practical relevance to consider instead a notion of
equivalence that guarantees the same answer sets even in combination with
other programs: Two logic programs $P, Q$ are \name{strongly
  equivalent}~\cite{LifschitzEtAl2001} if they can be exchanged in the context
of any other program~$R$ without affecting the answer sets of the overall
program. That is, $P$ and $Q$ are strongly equivalent if and only if for all
logic programs $R$ it holds that $P \cup R$ and $Q \cup R$ have the same
answer sets.

Although it has been known that strong equivalence of logic programs
under the stable model semantics can be translated into equivalence of
classical logical formulas, e.g.~\cite{lin:equivalence:2002}, developments in
the languages of answer set programming make this an issue of ongoing research~%
\cite{lpeq:2004,selp:2005,LifschitzEtAl2019,Heuer2020,heuer-wlp-2023,FandinnoLifschitz2023}.
The practical objective is to apply first-order provers to determine the
equivalence of
two
logic programs.

We now consider the situation where only a single program is given and a
strongly equivalent one is to be synthesized automatically. For the new
program the set of allowed predicates, including to some degree the position
within rules in which they are allowed, is restricted to a given vocabulary.
Not just ``absolute'' strong equivalence is of interest, but also strong
equivalence with respect to some background knowledge expressed as a logic
program. Thus, for given programs $P$ and $Q$, and vocabulary~$V$ we want to
find programs~$R$ in $V$ such that $P \cup R$ and $P \cup Q$ are strongly
equivalent.

Our question has two aspects: characterizing the \emph{existence} of a program $R$
for given $P,Q,V$ and, if one exists, the effective \emph{construction} of such
an~$R$. As we will show, existence can be addressed by Beth definability~%
\cite{beth:1953,craig:uses} on the basis of Craig interpolation~%
\cite{craig:linear} for first-order logic. The construction can then be
performed by extracting an interpolant from a proof of the first-order
characterization of existence. We realize this practically with the
first-order provers \CMProver~\cite{cw:pie:2016,cw:pie:2020} and \ProverN~%
\cite{prover9} and an interpolation technique for clausal tableaux~%
\cite{cw:interpolation:2021,cw:range:2023}.

To achieve this, we start from a known representation of logic programs in
classical first-order logic for the purpose of verifying strong equality. We
supplement it with a formal characterization to determine whether an arbitrary
given first-order formula represents a logic program, and, if so, to extract a
represented logic program from the formula. This novel ``reverse translation''
also has other potential applications in program transformation.

Beth definability and Craig interpolation play a key role for advanced tasks
in other fields of knowledge representation, in particular for query
reformulation in databases~%
\cite{segoufin:vianu:determinacy:05,nash:2010,toman:wedell:book,benedikt:book,bwp:pods:2023}
and description logics as well as ontology-mediated querying~%
\cite{tencate:etal:defdl:2005,tencate:etal:expressive:2013,toman:weddell:horn:2022,living:without:2023}.
Our work aims to provide for these lines of research the bridge to answer set
programming.

\textit{Structure of the Paper.} After providing in \cref{sec-background}
background material on strong equivalence as well as on interpolation and
definability we develop in \cref{sec-craig-beth} our technical results. Our
prototypical implementation\footnote{Available from
\url{http://cs.christophwernhard.com/pie/asp} as free software.} is then
described in \cref{sec-implementation}. We conclude in \cref{sec-conclusion}
with discussing related work and perspectives.
\newpage

\section{Background}
\label{sec-background}

\subsection{Notation}
\label{sec-notation}

We map between two formalisms: logic programs and formulas of classical
first-order logic without equality (briefly \name{formulas}). In both formalisms we have
\name{atoms} $p(t_1,\ldots,t_n)$, where $p$ is a \name{predicate} and
$t_1,\ldots,t_n$ are terms built from \name{functions}, including
\name{constants}, and \name{variables}. We assume variable names as case
insensitive to take account of the custom to write them uppercase in logic
programs and lowercase in first-order logic.
Predicates in logic programs are distinct from those in formulas, but
with a correspondence: If $p$ is a predicate for use in logic programs, then
the two predicates $p^0$ and $p^1$, both with the same arity as $p$, are for
use in formulas. Thus, predicates in formulas are always decorated with
a $0$ or $1$ superscript. To emphasize this, we sometimes speak of
\name{\superscripted formulas}.

A \defname{literal} is an atom or a negated atom.
A \defname{clause} is a disjunction of literals, a \defname{clausal formula} is
a conjunction of clauses. The empty disjunction is represented by $\false$, the
empty conjunction by $\true$. On occasion we write a clause also as an
implication.

A subformula occurrence in a formula has \defname{positive (negative) polarity}
if it is in the scope of an even (odd) number of possibly implicit negations.
A formula is \defname{universal} if occurrences of $\forall$ have only positive
polarity and occurrences of $\exists$ have only negative polarity.
Semantic entailment and equivalence of formulas are expressed by $\entails$ and $\equiv$.

Let $F$ be a formula. $\fun(F)$ is the set of functions occurring in it,
including constants, and $\pred(F)$ is the set of predicates occurring in it.
$\predpol(F)$ is the set of pairs $\la \mathit{pol}, p \ra$, where $p$ is a
predicate and $\mathit{pol} \in \{{+},{-}\}$, such that an atom with
predicate~$p$ occurs in $F$ with the polarity indicated by $\mathit{pol}$. We
write $\la +, p\ra$ and $\la -, p\ra$ succinctly as $+p$ and $-p$.
To map from the predicates occurring in a formula to predicates of logic
programs we define $\predplain(F) \eqdef \{p \mid p^i \in \pred(F), i \in
\{0,1\}\}$. For logic programs~$P$, we define $\fun(P)$ and $\pred(P)$
analogously as for formulas.

\subsection{Strong Equivalence as First-Order Equivalence}

We consider \name{disjunctive logic programs with negation in the head}~\cite[Sect. 5]{Lifschitz1996},
which provide a normal form for answer set programs~\cite{CabalarFerraris2007}.
A \defname{logic program} is a set of \defname{rules} of the form%
\[\begin{array}{l}
A_1 \por \ldots \por A_k  \por \pnot A_{k+1} \por \ldots \por \pnot A_l
\revimp
A_{l+1}, \ldots, A_{m}, \pnot A_{m+1}, \ldots, \pnot A_{n}.,
\end{array}
\] where
$A_1, \ldots, A_n$ are atoms, $0 \leq k \leq l \leq m \leq n$. Recall from \cref{sec-notation}
that an atom can have argument terms with functions and variables.
The \defname{positive/negative head} of a rule are the atoms $A_{1}, \ldots, A_{k}$
and $A_{k+1}, \ldots, A_{l}$ respectively.
Analogously the \defname{positive/negative body} of a rule are $A_{l+1}, \ldots, A_{m}$ and
$A_{m+1}, \ldots, A_{n}$ respectively. Answer sets with respect to the stable model
semantics for this class of programs are for example defined in~\cite{GebserEtAl2015}.

Next we review the definition of the translation~$\pfol$ used to express strong equivalence of
two logic programs in classical first-order logic.%
\footnote{With the notation~$\pfol$ we follow~\cite{FandinnoLifschitz2023}. In~\cite{Heuer2020,heuer-wlp-2023} $\sigma^{*}$ is used for the same translation.}
It makes use of
the fact that strong equivalence can be expressed in the intermediate logic of
here-and-there~\cite{LifschitzEtAl2001}, which in turn can be mapped to
classical logic~\cite{PearceEtAl2009}. For details and proofs we refer to~\cite{Heuer2020,heuer-wlp-2023,FandinnoLifschitz2023}.
Similar results appeared in~%
\cite{lin:equivalence:2002,PearceEtAl2009,ferrarisetal:circumscription:2011,PearceValverde2008}.
As stated in \cref{sec-notation} we assume for each program predicate $p$
two dedicated formula predicates $p^0$ and $p^1$ with the same arity.
If $A$ is an atom with predicate~$p$, then $A^0$ is $A$ with $p^0$ instead of
$p$, and $A^1$ is $A$ with $p^1$ instead of $p$.

\begin{defn}
  \label{def-gamma-and-sp}
  For a rule \[R = A_1 \por \ldots \por A_k \por \pnot A_{k+1} \por \ldots
  \por \pnot A_l
  \revimp A_{l+1}, \ldots, A_{m},
  \pnot A_{m+1}, \ldots, \pnot A_{n}.\] with variables~$\xs$
  define the first-order formulas $\rfolzero(R)$ and
  $\rfolone(R)$ as
  \begin{eqnarray*}
  \rfolzero(R) & \eqdef &
  \forall \xs\,
  (\textstyle\bigwedge_{i={l+1}}^{m} A^0_i \land
  \bigwedge_{i=m+1}^{n} \lnot A^1_i \imp
  \bigvee_{i=1}^{k} A^0_i \lor
  \bigvee_{i=k+1}^{l} \lnot A^1_i),\\
  \rfolone(R) & \eqdef &
  \forall \xs\,
  (\textstyle\bigwedge_{i={l+1}}^{m} A^1_i \land
  \bigwedge_{i=m+1}^{n} \lnot A^1_i \imp
  \bigvee_{i=1}^{k} A^1_i \lor
  \bigvee_{i=k+1}^{l} \lnot A^1_i).\\[-10pt]
  \end{eqnarray*}
  For a logic program~$P$ define the first-order formula~$\pfol(P)$ as
  \[\textstyle\pfol(P) \eqdef
  \bigwedge_{R \in P} \rfolzero(R) \land
  \bigwedge_{R \in P} \rfolone(R).\]
  and define the first-order formula~$\SP_P$ as
  \[\textstyle\SP_{P} \eqdef \bigwedge_{p \in \pred(P)} \forall \xs\, (p^{0}(\xs) \imp p^{1}(\xs)),\]
  where variables $\xs$ match the arity of $p$.
\end{defn}

Using the transformation~$\pfol$ and the formula~$\SP_{P}$ we can express strong
equivalence as an equivalence in first-order logic.

\begin{prop}[\cite{Heuer2020}]\label{prop-strong-equiv-fo}
  Under the stable model semantics two logic programs $P$ and $Q$ are strongly
  equivalent iff the following equivalence holds in classical
  first-order logic $\SP_{P \cup Q} \land \pfol(P) \equiv \SP_{P \cup Q} \land \pfol(Q)$.
\end{prop}

\subsection{Definition Synthesis with Craig Interpolation}

A formula $Q(\xs)$ is \defname{implicitly definable} in terms of a vocabulary
(set of predicate and function symbols)~$V$ within a sentence~$F$ if, whenever
two models of $F$ agree on values of symbols in $V$, then they agree on the
extension of $Q$, i.e., on the tuples of domain members that satisfy $Q$. In
the special case where $Q$ has no free variables, this means that they agree
on the truth value of $Q$. Implicit definability can be expressed as
\begin{equation}\label{eq-defimp}
F \land F' \entails \forall \xs\, (Q(\xs) \equi Q'(\xs)),
\end{equation}
where $F'$ and $Q'$ are copies of $F$ and $Q$ with all symbols not in $V$
replaced by fresh symbols. This semantic notion contrasts with the
syntactic one of \defname{explicit definability}: A formula $Q(\xs)$ is
\name{explicitly definable} in terms of a vocabulary~$V$ within a sentence~$F$
if there exists a formula $R(\xs)$ in the vocabulary~$V$ such that
\begin{equation}\label{eq-defexp}
  F \entails \forall \xs\, (R(\xs) \equi Q(\xs)).
\end{equation}
The ``Beth property''~\cite{beth:1953} states equivalence of both notions and
applies to first-order logic.
``Craig interpolation'' is a tool that can be applied to prove the ``Beth
property''~\cite{craig:uses}, and, moreover to construct formulas~$R$ from
given $F,Q,V$ from a proof of implicit definability. Craig's interpolation
theorem~\cite{craig:linear} applies to first-order logic and states that if a
formula~$F$ entails a formula~$G$ (or, equivalently, that $F \imp G$ is
valid), then there exists a formula~$H$, a \defname{Craig interpolant} of $F$
and~$G$, with the properties that $F \entails H$, $H \entails G$, and the
vocabulary of $H$ (predicate and function symbols as well as free variables)
is in the common vocabulary of $F$ and $G$. Craig's theorem can be
strengthened to the existence of \defname{Craig-Lyndon interpolants}~%
\cite{lyndon} that satisfy the additional property that predicates in $H$
occur only in polarities in which they occur in both $F$ and~$G$.
In our technical framework this condition is expressed as $\predpol(H)
\subseteq \predpol(F) \cap \predpol(G)$.

Craig's interpolation theorem can be proven by constructing $H$ from a proof
of $F \entails G$. This works on the basis of sequent calculi~%
\cite{smullyan:book:68,takeuti:1987} and analytic tableaux~%
\cite{fitting:book}. For calculi from automated first-order reasoning various
approaches have been considered~%
\cite{slagle:interpolation:1970,huang:95,bonacina:15:on,kovacs:17,cw:interpolation:2021,cw:range:2023}.
A method~\cite{cw:interpolation:2021} for clausal tableaux~\cite{letz:habil}
performs Craig-Lyndon interpolation and operates on proofs emitted by a
general first-order prover, without need to modify the prover for interpolation,
and inheriting completeness for full first-order logic from it.
Indirectly that method also works on resolution proofs expanded into
trees~\cite{cw:range:2023}.

Observe that the characterization of implicit definability (\ref{eq-defimp})
can also be expressed as $F \land Q(\xs) \entails F' \imp Q'(\xs)$. An
``explicit'' definiens $R(\xs)$ can now be constructed from given $F$,
$Q(\xs)$ and $V$ just as a Craig interpolant of $F \land Q(\xs)$ and $F' \imp
Q'(\xs)$.
The synthesis of definitions by Craig interpolation was recognized as a
logic-based core technique for view-based query reformulation in relational
databases~%
\cite{segoufin:vianu:determinacy:05,nash:2010,toman:wedell:book,benedikt:book}.
Often strengthened variations of Craig-Lyndon interpolation are used there
that preserve criteria for domain independence, e.g., through relativized
quantifiers~\cite{benedikt:book} or range-restriction \cite{cw:range:2023}.

\section{Variations of Craig Interpolation and Beth Definability for Logic Programs}
\label{sec-craig-beth}

We synthesize logic programs according to a variation of Beth's theorem
justified by a variation of Craig-Lyndon interpolation. Craig-Lyndon
interpolation ``out of the box'' is not sufficient for our purposes: To obtain
logic programs as definientia we need as a basis a stronger form of
interpolation where the interpolant is not just a first-order formula but,
moreover, is the first-order encoding of a logic program, permitting to actually
extract the program.

\subsection{Extracting Logic Programs from a First-Order Encoding}
\label{sec-extracting}

We address the questions of how to abstractly characterize first-order
formulas that encode a logic program, and how to extract the program from a
first-order formula that meets the characterization. As specified in
\cref{sec-notation}, we assume first-order formulas over predicates
superscripted with $0$ and $1$. The notation $\SP_P$ (\cref{def-gamma-and-sp})
is
now overloaded for \superscripted formulas~$F$ as $\SP_F \eqdef \bigwedge_{p
  \in \predplain(F)} \forall \xs\, (p^0(\xs) \imp p^1(\xs))$, where
variables $\xs$ match the arity of $p$.
We introduce a convenient notation for the result of systematically renaming
all $0$-superscripted predicates to their $1$-superscripted correspondence.
\begin{defn}\label{def-ren}
  For \superscripted first-order formulas $F$ define $\ren(F)$ as $F$ with all
  occurrences of $0$-superscripted predicates $p^0$ replaced by the
  corresponding $1$-superscripted predicate~$p^1$.
\end{defn}
Obviously $\ren$ can be moved over logic operators, e.g., $\ren(F \land G)
\equiv \ren(F) \land \ren(G)$, and $\ren(\forall \xs\, F) \equiv \forall \xs\,
\ren(F)$.
Semantically, $\ren$ preserves entailment and thus also equivalence.
\begin{prop}
  For \superscripted first-order formulas $F$ and $G$ it holds that
  \subpropinline{prop-ren-entails} If $F \entails G$, then $\ren(F) \entails
  \ren(G)$.
  \subpropinline{prop-ren-equiv} If $F \equiv G$, then $\ren(F) \equiv \ren(G)$.
\end{prop}
Observe that for all rules $R$ it holds that $\rfolone(R) =
\ren(\rfolzero(R))$ and for all logic programs~$P$ it holds that $\pfol(P)
\entails \ren(\pfol(P))$. On the basis of $\ren$ we define a first-order criterion for
a formula encoding a logic program, that is, a first-order entailment that holds if and
only if a given first-order formula encodes a logic program.
\begin{defn}
  \label{def-enc-lp}
  A \superscripted first-order formula $F$ is said to \defname{encode a logic
    program} iff $F$ is universal and $\SP_{F} \land F \entails \ren(F).$
\end{defn}
This criterion adequately characterizes that the formula represents a logic
program on the basis of the translation $\pfol$. The following theorem
justifies this.
\begin{thm}[Formulas Encoding a Logic Program]
  \subthminline{thm-p-elp} For all logic programs $P$ it holds that $\pfol(P)$
  is a \superscripted first-order formula that encodes a logic program.
  \subthminline{thm-elp-p} If a \superscripted first-order formula $F$ encodes a
  logic program, then there exists a logic program $P$ such that $\SP_F
  \entails \pfol(P) \equi F$, $\pred(P) \subseteq \predplain(F)$ and $\fun(P)
  \subseteq \fun(F)$. Moreover, such a program~$P$ can be effectively
  constructed from $F$.
\end{thm}
\begin{proof}
  (\ref{thm-p-elp}) Immediate from the definition of $\pfol$.
  (\ref{thm-elp-p}) \cref{proc-elp-p} specified below and proven correct in
  \cref{prop-correct-proc-elp-p} shows the construction of a
  suitable program. \qed
\end{proof}
\Cref{thm-elp-p} claims that the vocabulary of the program is only
\emph{included} in the respective vocabulary of the formula. This gives for the
formula the freedom of a larger vocabulary with symbols that may be
eliminated, e.g., by simplifications.

\begin{proc}[Decoding an Encoding of a Logic Program]
  \label{proc-elp-p}
  \ \\
  \algoinput  A \superscripted first-order formula~$F$ that encodes a logic
  program.

  \algomethod
  \begin{enumerate}
  \item \label[stp]{item-proc-elp-p-1} Bring $F$ into conjunctive normal form
    matching
    $\forall \xs\, (M_0 \land M_1)$, where $M_0$ and $M_1$ are clausal formulas such
    that in all clauses of $M_0$ there is a literal whose predicate has
    superscript $0$ and in all clauses of $M_1$ the predicates of all literals
    have superscript $1$.
  \item \label[stp]{item-proc-elp-p-2} Partition $M_1$ into two clausal formulas $M_1'$
    and $M_1''$ such that
    \[\forall \xs\, \ren(M_0) \entails \forall \xs\, M_1''.\] A possibility is
    to take $M_1' = M_1$ and $M_1'' = \true$. Another possibility that often
    leads to a smaller $M'$ is to consider each clause $C$ in $M_1$ and place it into
    $M_1''$ or $M_1'$ depending on whether there is a clause $D$ in $M_0$ such that
    $\ren(D)$ subsumes $C$.
  \item Let $P$ be the set of rules
    \[A_1 \por \ldots \por A_k \por \pnot A_{k+1} \por \ldots
    \por \pnot A_l
    \revimp A_{l+1}, \ldots, A_{m},
    \pnot A_{m+1}, \ldots, \pnot A_{n}\]
    for each clause $\bigwedge_{i={l+1}}^{m}  A^0_i \land
    \bigwedge_{i=m+1}^{n}\! \lnot A^1_i \imp
    \bigvee_{i=0}^{k} A^0_i \lor
    \bigvee_{i=k+1}^{l}\! \lnot A^1_i$
    in $M_0 \land M_1'$.
  \end{enumerate}

  \algooutput Return $P$, a logic program such that $\SP_P \entails
  \pfol(P) \equi F$, $\pred(P) \subseteq \predplain(F)$ and
  $\fun(P) \subseteq \fun(F)$.
\end{proc}

\begin{prop}
  \label{prop-correct-proc-elp-p}
  \cref{proc-elp-p} is correct.
\end{prop}

\begin{proof}
\prlReset{proc-elp-p}

For an input formula $F$ and output program $P$ the syntactic requirements
$\pred(P) \subseteq \predplain(F)$ and $\fun(P) \subseteq \fun(F)$ are easy to
see from the construction of $P$ by the procedure. To prove the
semantic requirement $\SP_F \entails \pfol(P) \equi F$ we first note the
following assumptions, which follow straightforwardly from the specification
of the procedure.
\begin{enumerate}[left=\prenum]\prooftab
  \item[\prl{a1}] $\SP_{F} \land F \entails \ren(F)$.
  \item[\prl{a2}] $F \equiv \forall \xs\, (M_0 \land M_1' \land M_1'')$.
  \item[\prl{a3}] $\ren(\forall \xs\, M_0) \entails \forall \xs\, M_1''$.
  \item[\prl{a4}] $\pfol(P) \equiv \forall \xs\, (M_0 \land \ren(M_0) \land M_1')$.
\end{enumerate}
The semantic requirement $\SP_F \entails \pfol(P) \equi F$ can then be derived
as follows.
\begin{enumerate}[left=\prenum]\prooftab
  \item[\prl{s3}] $\SP_{F} \land F \entails \ren(\forall \xs\, M_0)$.
    \hfill By \pref{a2} and \pref{a1}.
  \item[\prl{s4}] $\SP_{F} \land F \entails
    \forall \xs\, (M_0 \land \ren(M_0) \land M_1')$.
    \hfill By \pref{s3} and \pref{a2}.
  \item[\prl{s5}] $\SP_{F} \land F \entails \pfol(P)$.
    \hfill By \pref{s4} and \pref{a4}.
  \item[\prl{s1}] $\pfol(P) \entails \forall \xs\, (M_0 \land M_1'' \land M_1')$.
    \hfill By \pref{a4} and \pref{a3}.
  \item[\prl{s2}] $\pfol(P) \entails F$.
    \hfill By \pref{s1} and \pref{a2}.
  \item[\prl{s6}] $\SP_F \entails \pfol(P) \equi F$.
    \hfill By \pref{s2} and \pref{s5}.
\end{enumerate}
\qed
\end{proof}
Some examples illustrate \cref{proc-elp-p} and the \name{encoding a
  logic program} property.\hspace{-1em}
\begin{examp}
  \ \\
  \subexamp{ex-enc-0} Consider the following clauses and programs.

  \noindent
  $\begin{array}[t]{ll}
    C_1 = & \lnot \fp^0 \lor \fq^1 \lor \fr^0\\
    C_2 = & \lnot \fs^1 \lor \ft^1 \lor \fu^1\\
    C_3 = & \lnot \fp^1 \lor \fq^1 \lor \fr^1\\
  \end{array}$
  \hfill
  $\begin{array}[t]{ll}
    P = & \fr \revimp \fp, \pnot \fq.\\
        & \pnot \fs \revimp \pnot \ft, \pnot \fu.\\
        & \pnot \fp \revimp \pnot \fq, \pnot \fr.\\
  \end{array}$
  \hfill
  $\begin{array}[t]{ll}
    P' = & \fr \revimp \fp, \pnot \fq.\\
         & \pnot \fs \revimp \pnot \ft, \pnot \fu.\\
  \end{array}$

  \smallskip

  \noindent
  Assume as input of \cref{proc-elp-p} the formula $F = C_1 \land C_2
  \land C_3$. Then $M_0 = C_1$ and $M_1 = C_2 \land C_3$. If in
  \cref{item-proc-elp-p-2} we set $M_1' = M_1$ and $M_1'' = \true$, then
  the extracted program is $P$. We might, however, also set
  $M_1' = C_2$ and $M_1'' = C_3$ and obtain the shorter strongly equivalent
  program $P'$.

  \subexamp{ex-enc-5} For $F = \lnot \fp^1 \lor \fq^1 \lor \fr^0$ we have to
  set $M_1' = M_1'' = M_1 = \true$ and the procedure yields $\{\fr \por \pnot
  \fp \revimp \pnot \fq.\}$.

  \subexamp{ex-enc-3} The formula $F = \lnot \fp^0 \lor \fq^1 \lor \fr^0$ does
  not encode a logic program because $\SP_{F} \land F \not\entails \ren(F)$.
\end{examp}

\begin{remark}\label{remark-simp}
For the extracted program it is desirable that it is not unnecessarily large.
Specifically it should not contain rules that are easily identified as
redundant. \Cref{item-proc-elp-p-2} of \cref{proc-elp-p} permits
techniques to keep $M'$ small. Other possibilities include well-known formula
simplification techniques that preserve equivalence such as removal of
tautological or subsumed clauses and may be integrated into clausification,
\cref{item-proc-elp-p-1} of the procedure.
In addition, conversions that just preserve equivalence modulo $\SP_F$ may be
applied, conceptually as a preprocessing step, although in practice possibly
implemented on the clausal form. \cref{proc-elp-p} then receives as
input not $F$ but a universal first-order formula $F'$ whose vocabulary is
included in that of $F$ with the property
\begin{equation}
  \label{eq-simp-sp}
  \SP_F \entails F' \equi F.
\end{equation}
Formula $F'$ then also encodes a program: That $\SP_{F'} \land F' \entails
\ren(F')$ follows from $\SP_{F} \land F \entails \ren(F)$, (\ref{eq-simp-sp})
and \cref{prop-ren-equiv}. \cref{proc-elp-p} guarantees for its output
$\SP_{F'} \entails \pfol(P) \equi F'$, hence by (\ref{eq-simp-sp}) it follows
that $\SP_F \entails \pfol(P) \equi F$.
\end{remark}

\begin{examp}
  \label{examp-simp}
  Consider the following clauses and programs.

  \noindent
  $\begin{array}[t]{ll}
    C_1 = & \lnot \fp^0 \lor \fq^1 \lor \fr^0\\
    C_2 = & \lnot \fp^0 \lor \fq^1 \lor \fr^1\\
    C_3 = & \lnot \fp^1 \lor \fq^1 \lor \fr^1\\
  \end{array}$
  \hfill
  $\begin{array}[t]{ll}
    P = & \fr \revimp \fp, \pnot \fq.\\
    & \revimp \fp, \pnot \fq, \pnot \fr.
  \end{array}$
  \hfill
  $\begin{array}[t]{ll}
    P' = & \fr \revimp \fp, \pnot \fq.\\
  \end{array}$

  \smallskip
  \noindent
  Assume as input of \cref{proc-elp-p} the formula $F = C_1 \land C_2
  \land C_3$. Then $M_0 = C_1 \land C_2$ and $M_1 = C_3$, and, aiming at a
  short program, we can set $M_1' = \true$ and $M_1'' = C_3$. The extracted
  program is then~$P$. By preprocessing $F$ according to
  \cref{remark-simp} we can eliminate $C_2$ from $F$ and obtain the
  shorter strongly equivalent program~$P'$.
\end{examp}

\subsection{A Refinement of Craig Interpolation for Logic Programs}
\label{sec-craig}

With the material from \cref{sec-extracting} on extracting logic programs
from formulas we can now state a theorem on \name{LP-interpolation}, where
\name{LP} stands for \name{logic program}. It is a variation of Craig
interpolation applying to first-order formulas that encode logic programs. The
theorem states not only the existence of an LP-interpolant, but, moreover, also
claims effective construction.
\begin{thm}[LP-Interpolation]
  \label{thm-ipol}
  Let $F$ be a \superscripted first-order formula that encodes a logic program
  and let $G$ be a \superscripted first-order formula such that $\fun(F)
  \subseteq \fun(G)$ and $\SP_{F} \land F \entails \SP_{G} \imp G$. Then there
  exists a \superscripted first-order formula $H$, called the
  \defname{LP-interpolant} of $F$ and~$G$, such that
  \begin{enumerate}
  \item \label[clm]{item-thm-ipol-1} $\SP_{F} \land F \entails H$.
  \item \label[clm]{item-thm-ipol-2} $H \entails \SP_{G} \imp G$.
  \item \label[clm]{item-thm-ipol-3} $\predpol(H) \subseteq S \cup \{+p^1 \mid +p^0 \in S\} \cup \{-p^1 \mid -p^0 \in S\}$, where
        $S = \predpol(\SP_{F} \land F) \cap \predpol(\SP_G \imp G)$.
  \item \label[clm]{item-thm-ipol-4} $\fun(H) \subseteq \fun(F)$.
  \item \label[clm]{item-thm-ipol-5} $H$ encodes a logic program.
  \end{enumerate}
  Moreover, if existence holds, then an LP-interpolant $H$ can be effectively
  constructed, via a universal Craig-Lyndon interpolant of $F \land \SP_{F}$
  and $\SP_{G} \imp G$.
\end{thm}
\begin{proof}
  \prlReset{thm-ipol} We show the construction of a suitable formula $H$. Let
  $H'$ be a Craig-Lyndon interpolant of $\SP_{F} \land F$ and $\SP_{G} \imp
  G$. Since $F$ and $\SP_F$ are universal first-order formulas and we have the
  precondition $\fun(F) \subseteq \fun(G)$, we may in addition assume that
  $H'$ is a universal first-order formula. (This additional condition is
  guaranteed for example by the interpolation method from
  \cite{cw:interpolation:2021}, which computes $H'$ directly from a clausal
  tableaux proof, or indirectly from a resolution proof~\cite{cw:range:2023}.)
  Define $H \eqdef H' \land \ren(H')$.
  \Cref{item-thm-ipol-2,item-thm-ipol-4,item-thm-ipol-5}
  of the theorem statement are then easy to see.
  \Cref{item-thm-ipol-1} can be shown as follows. We may
  assume the following.
  \begin{enumerate}[left=\prenum]\prooftab
  \item[\prl{a1}] $\SP_{F} \land F \entails H'$.
    \hfill Since $H'$ is a Craig-Lyndon interpolant.
  \item[\prl{a2}] $\SP_{F} \land F \entails \ren(F)$.
    \hfill Since $F$ encodes a logic program.
  \end{enumerate}
  \Cref{item-thm-ipol-1} can then be derived in the following steps.
  \begin{enumerate}[left=\prenum]\prooftab
  \item[\prl{s1}] $\ren(\SP_{F} \land F) \entails \ren(H')$.
    \hfill By \pref{a1} and \abbref{prop-ren-entails}.
  \item[\prl{s2}]$\SP_{F} \land F \entails \ren(\SP_{F} \land F)$.
    \hfill By \pref{a2}, since $\ren(\SP_{F}) \equiv \true$.
  \item[\prl{s3}] $\SP_{F} \land F \entails \ren(H')$.
    \hfill By \pref{s2} and \pref{s1}.
  \item[\prl{s4}] $\SP_{F} \land F \entails H$.
    \hfill By \pref{s3} and \pref{a1}, since $H = H' \land \ren(H')$.
  \end{enumerate}
  \Cref{item-thm-ipol-3} follows because since $H'$ is a Craig-Lyndon
  interpolant it holds that $\predpol(H') \subseteq \predpol(\SP_{F} \land F)
  \cap \predpol(\SP_G \imp G)$. With the predicate occurrences in $\ren(H')$,
  i.e., $1$-superscripted predicates in positions of $0$-superscripted
  predicates in $H'$, we obtain the restriction of $\predpol(H)$ stated as
  \cref{item-thm-ipol-3}.
  \qed
\end{proof}

For an LP-interpolant of formulas $F$ and $G$, where $F$ encodes a logic
program, the semantic properties stated in
\cref{item-thm-ipol-1,item-thm-ipol-2} are those of a Craig or Craig-Lyndon
interpolant of $\SP_{F} \land F$ and $\SP_{G} \imp G$. The allowed
polarity/predicate pairs are those common in $\SP_{F} \land F$ and $\SP_{G}
\imp G$, as in a Craig-Lyndon interpolant, and, in addition, the
$1$-superscripted versions of polarity/predicate pairs that appear only
$0$-superscripted among these common pairs. These additional pairs are those
that might occur in the result of $\ren$ applied to a Craig-Lyndon
interpolant. In contrast to a Craig interpolant, functions are only
constrained by the first given formula~$F$. Permitting only functions
common to $F$ and $G$ can result in an interpolant with existential
quantification, which thus does not encode a program. \Cref{item-thm-ipol-5}
states that the LP-interpolant indeed encodes a logic program as characterized
in \cref{def-enc-lp}. This property is, so-to-speak, passed through from the
given formula $F$ to the LP-interpolant.

\subsection{Effective Projective Definability of Logic Programs}
\label{sec-definability}

We present a variation of the ``Beth property'' that applies to logic programs
with stable model semantics and takes strong equivalence into account. The
underlying technique is our LP-interpolation, \cref{thm-ipol}. It maps into
Craig-Lyndon interpolation for first-order logic, utilizing that strong
equivalence of logic programs can be expressed as first-order equivalence of
encoded programs. This approach allows the effective construction of logic
programs in the role of ``explicit definitions'' via Craig-Lyndon
interpolation on first-order formulas. Our variation of the ``Beth property''
is \name{projective} as it is with respect to a given set of predicates
allowed in the definiens.
While our LP-interpolation theorem was expressed in terms of first-order
formulas that encode logic programs, we now phrase definability entirely in
terms of logic programs.\footnote{LP-interpolation could also be phrased in
terms of logic programs, providing an interpolation result for logic programs
on its own, not just as basis for definability. We plan to address this in
future work.}

\begin{thm}[Effective Projective Definability of Logic Programs]
  \label{thm-def}
  Let $P$ and $Q$ be logic programs and let $V \subseteq \pred(P) \cup
  \pred(Q)$ be a set of predicates. The existence of a logic program $R$ with
  $\pred(R) \subseteq V$, $\fun(R) \subseteq \fun(P) \cup \fun(Q)$ such that
  $P \cup R$ and $P \cup Q$ are strongly equivalent is expressible as
  entailment between two first-order formulas. Moreover, if existence holds,
  such a program~$R$ can be effectively constructed, via a universal
  Craig-Lyndon interpolant of the left and the right side of the entailment.
\end{thm}

\begin{proof}
  \prlReset{thm-def} The first-order entailment that characterizes the
  existence of a logic program $R$ is $\SP_{P} \land \SP_{Q} \land \pfol(P)
  \land \pfol(Q) \entails \lnot \SP_{P'} \lor \lnot \SP_{Q'} \lor \lnot \pfol(P')
  \lor \pfol(Q')$, where the primed programs $P'$ and $Q'$ are like $P$ and
  $Q$, except that predicates not in~$V$ are replaced by fresh predicates. If
  the entailment holds, we can construct a program $R$ as follows: Let
  $H$ be the LP-interpolant of $\pfol(P) \land \pfol(Q)$ and $\lnot \pfol(P')
  \lor \pfol(Q')$, as specified in \cref{thm-ipol}, and extract the
  program $R$ from $H$ with \cref{proc-elp-p}.
  That $R$ constructed in this way has the properties claimed in the theorem
  statement can be shown as follows. Since $H$ is an LP-interpolant it follows
  that
  \begin{enumerate}[left=\prenum]\prooftab
  \item[\prl{a1}]
    $\SP_{P} \land \SP_{Q} \land \pfol(P) \land \pfol(Q) \entails H$.
  \item[\prl{a2}] $H \entails \lnot \SP_{P'} \lor \lnot \SP_{Q'}
    \lor \lnot \pfol(P') \lor \pfol(Q')$.
  \item[\prl{a3}] $\predplain(H) \subseteq V$.
  \item[\prl{a4}] $\fun(H) \subseteq \fun(P) \cup \fun(Q)$.
  \item[\prl{a5}] $H$ encodes a logic program.
  \end{enumerate}
  From the preconditions of the theorem and since $R$ is
  extracted from $H$ with \cref{proc-elp-p} and thus
  meets the properties stated in \cref{thm-elp-p}
  it follows that
  \begin{enumerate}[left=\prenum]\prooftab
  \item[\prl{b1}] $V \subseteq \pred(P) \cup \pred(Q)$.
  \item[\prl{b2}] $\SP_H \entails \pfol(R) \equi H$.
  \item[\prl{b3}] $\pred(R) \subseteq \predplain(H)$.
  \item[\prl{b4}] $\fun(R) \subseteq \fun(H)$.
  \end{enumerate}
  The claimed properties of the theorem statement can then be derived as steps
  \pref{s1}, \pref{s2} and \pref{s8} as follows.
  \begin{enumerate}[left=\prenum]\prooftab
  \item[\prl{s1}] $\pred(R) \subseteq V$.
    \hfill By \pref{b3} and \pref{a3}.
  \item[\prl{s2}] $\fun(R) \subseteq \fun(P) \cup \fun(Q)$.
    \hfill By \pref{b4} and \pref{a4}.
  \item[\prl{s3}] $\SP_P \land \SP_Q \entails \SP_H$.
     \hfill By \pref{b1} and \pref{a3}.
  \item[\prl{s4}] $\SP_{P} \land \SP_{Q} \land \pfol(P) \land \pfol(Q) \entails
    \pfol(R)$.
    \hfill By \pref{s3}, \pref{b2} and \pref{a1}.
  \item[\prl{s4a}]
    $\SP_{P'} \land \SP_{Q'} \land \pfol(P') \land \lnot \pfol(Q') \entails
    \lnot H$.
    \hfill By \pref{a2}.
  \item[\prl{s5}]
    $\SP_{P} \land \SP_{Q} \land \pfol(P) \land \lnot \pfol(Q) \entails
    \lnot H$.
    \hfill By \pref{s4a}, \pref{a3} and \pref{a4}.
  \item[\prl{s6}]
    $\SP_{P} \land \SP_{Q} \land \pfol(P) \land \lnot \pfol(Q) \entails
    \lnot \pfol(R)$.
    \hfill By \pref{s5}, \pref{s3} and \pref{b2}.
  \item[\prl{s7}] $\SP_{P} \land \SP_{Q} \land \pfol(P) \land \pfol(R)
    \equiv \SP_{P} \land \SP_{Q} \land \pfol(P) \land \pfol(Q)$.
    \hfill By \pref{s6} and \pref{s4}.
  \item[\prl{s8}] $P \cup R$ and $P \cup Q$ are strongly equivalent.
    \hfill By \pref{s7} and \abbref{prop-strong-equiv-fo}.
  \end{enumerate}

  Conversely, we have to show that if there exists a logic program $R$ with
  the properties in the above theorem statement, then the characterizing entailment
  $\SP_{P} \land \SP_{Q} \land \pfol(P) \land \pfol(Q) \entails \lnot \SP_{P'}
  \lor \lnot \SP_{Q'} \lor \lnot \pfol(P') \lor \pfol(Q')$ does hold. We may
  assume
  \begin{enumerate}[left=\prenum]\prooftab
  \item[\prl{ca1}] $\pred(R) \subseteq V \subseteq \pred(P) \cup \pred(Q)$.
  \item[\prl{ca2}] $\fun(R) \subseteq \fun(P) \cup \fun(Q)$.
  \item[\prl{ca3}] $P \cup R$ and $P \cup Q$ are strongly equivalent.
  \end{enumerate}
  The characterizing entailment can then be derived as follows.
  \begin{enumerate}[left=\prenum]\prooftab
  \item[\prl{t0}] $\SP_{P} \land \SP_{Q} \entails \SP_{R}$
    \hfill By \pref{ca1}.
  \item[\prl{t1}] $\SP_{P} \land \SP_{Q} \land \pfol(P) \land \pfol(R)
    \equiv \SP_{P} \land \SP_{Q} \land \pfol(P) \land \pfol(Q)$.
    \hfill By \pref{ca3}, \abbref{prop-strong-equiv-fo} and \pref{t0}.
  \item[\prl{t2}] $\SP_{P} \land \SP_{Q} \land \pfol(P) \land \pfol(Q) \entails
    \pfol(R)$.
    \hfill By \pref{t1}.
  \item[\prl{t3}]
    $\SP_{P} \land \SP_{Q} \land \pfol(P) \land \lnot \pfol(Q) \entails
    \lnot \pfol(R)$.
    \hfill By \pref{t1}.
  \item[\prl{t4}]
    $\SP_{P'} \land \SP_{Q'} \land \pfol(P') \land \lnot \pfol(Q') \entails
    \lnot \pfol(R)$.
    \hfill By \pref{t3}, \pref{ca1} and \pref{ca2}.
  \item[\prl{t5}]
    $\SP_{P} \land \SP_{Q} \land \pfol(P) \land \pfol(Q) \entails
    \lnot \SP_{P'} \lor \lnot \SP_{Q'} \lor \lnot \pfol(P') \lor \pfol(Q') $.
    \hfill By \pref{t4} and \pref{t2}.
  \end{enumerate}
\qed
\end{proof}

We now give some examples from the application point of view.

\begin{examp}
  \label{examp-thm-def}
  The following examples show for given programs $P,Q$ and sets~$V$ of
  predicates a possible value of $R$ according to \cref{thm-def}.

  \smallskip
  \subexamp{examp-thm-def-red} %
  $\begin{array}[t]{ll}
    Q = & \fp \revimp \fq, \fr.\\
    & \fp \por \fq \revimp \fr.\\
    & \fq \revimp \fq, \fs.
  \end{array}$
  \hfill
  $\begin{array}[t]{ll}
    V = & \{\fp,\fr\}\\
  \end{array}$
  \hfill
  $\begin{array}[t]{ll}
    R = & \fp \revimp \fr.\\
  \end{array}$

  \noindent
  In this first example we consider the special case where $P$ is empty
  and thus not shown. Predicates $\fq$ and $\fs$ are redundant in $Q$,
  ``absolutely'' and not just relative to a program~$P$. By
  \cref{thm-def}, this is proven with the characterizing first-order
  entailment and, moreover, a strongly equivalent reformulation of $Q$ without
  $\fq$ and $\fs$ is obtained as $R$.

  \smallskip
  \subexamp{examp-thm-def-one}  %
  $\begin{array}[t]{ll}
    P = & \fp(X) \revimp \fq(X).
  \end{array}$
  \hfill
  $\begin{array}[t]{ll}
    Q = & \fr(X) \revimp \fp(X).\\
    & \fr(X) \revimp \fq(X).
  \end{array}$
  \hfill
  $\begin{array}[t]{ll}
    V = & \{\fp,\fr\}\\
  \end{array}$
  \hfill
  $\begin{array}[t]{ll}
    R = & \fr(X) \revimp \fp(X).\\
  \end{array}$

  \noindent
  Only $\fr$ and $\fp$ are allowed in $R$. Or, equivalently, $\fq$ is
  redundant in $Q$, \emph{relative} to program~$P$. Again, this is proven with
  the characterizing first-order entailment and, moreover, a strongly
  equivalent reformulation of $Q$ without $\fq$ is obtained as $R$.
  It is the clause in $Q$ with $\fq$ that is redundant relative to $P$ and
  hence is eliminated in $R$.

  \smallskip
  \subexamp{examp-thm-def-two}  %
  $\begin{array}[t]{ll}
    P = & \revimp \fp(X), \fq(X).
  \end{array}$
  \hfill
  $\begin{array}[t]{ll}
    Q = & \fr(X) \revimp \fp(X), \pnot \fq(X).\\[1ex]
    R = & \fr(X) \revimp \fp(X).
  \end{array}$
  \hfill
  $\begin{array}[t]{ll}
    V = & \{\fp,\fr\}\\
  \end{array}$

  \noindent
  Only $\fr$ and $\fp$ are allowed in $R$. The negated literal with $\fq$ in
  the body of the rule in $Q$ is redundant relative to $P$ and is eliminated
  in $R$.

  \smallskip
  \subexamp{examp-thm-def-map}  %
  $\begin{array}[t]{ll}
    P = & \fp(X) \revimp \fq(X), \pnot \fr(X).\\
    & \fp(X) \revimp \fs(X).\\
    & \pnot \fr(X) \por \fs(X) \revimp \fp(X).\\
    & \fq(X) \por \fs(X) \revimp \fp(X).
  \end{array}$
  \hfill
  $\begin{array}[t]{ll}
    Q = & \ft(X) \revimp \fp(X).\\[1ex]
    R = & \ft(X) \revimp \fq(X), \pnot \fr(X).\\
        & \ft(X) \revimp \fs(X).
  \end{array}$
  \hspace{-20pt}
  $\begin{array}[t]{ll}
    V = & \{\fq,\fr,\fs,\ft\}\\
  \end{array}$

  \noindent
  The predicate~$\fp$ is not allowed in $R$. The idea is that $\fp$ is a
  predicate that can be used by a client but is not in the actual
  knowledge base. Program~$P$ expresses a schema mapping from the client
  predicate $\fp$ to the knowledge base predicates $\fq,\fr,\fs$. The result
  program~$R$ is a rewriting of the client query~$Q$ in terms of knowledge
  base predicates. Only the first two rules of $P$ actually describe the
  mapping. The other two rules complete them to a full
  definition, similar to Clark's completion~\cite{clark:1978}, but here
  yielding also a program. Such completed predicate specifications seem
  necessary for the envisaged reformulation tasks.

  \smallskip
  \subexamp{examp-thm-def-map-fold}  %
  $\begin{array}[t]{ll}
    P = & \text{As in \cref{examp-thm-def-map}.}
  \end{array}$
  \hspace{10pt}%
  $\begin{array}[t]{ll}
    Q = & \ft(X) \revimp \fq(X), \pnot \fr(X).\\
    & \ft(X) \revimp \fs(X).\\[1ex]
    R = & \ft(X) \revimp \fp(X).\\
  \end{array}$
  \hfill
  $\begin{array}[t]{ll}
    V = & \{\fp,\ft\}\\
  \end{array}$

  \noindent
  In this example $P$ is from \cref{examp-thm-def-map}, while $Q$ and
  $R$ are also from that example but switched. The vocabulary allows only
  $\fp$ and $\ft$. While \cref{examp-thm-def-map} realizes
  an unfolding of $\fp$, this example realizes folding into $\fp$.

  \smallskip
  \subexamp{examp-thm-def-recur}  %
  $\begin{array}[t]{ll}
    P = & \fn(X) \revimp \fz(X).\\
        & \fn(X) \revimp \fn(Y), \fs(Y,X).\\[1ex]
    V = & \{\fz,\fs\}\\
  \end{array}$
  \hfill
  $\begin{array}[t]{lll}
    Q = & \pnot \fn(X_2) \revimp & \fz(X_0), \fs(X_0,X_1), \\ & & \fs(X_1,X_2).\\[1ex]
    R = & \multicolumn{2}{l}{\revimp \fz(X_0), \fs(X_0,X_1), \fs(X_1,X_2).}
  \end{array}$

  \noindent Program $P$ defines natural numbers recursively. Program~$Q$
  has a rule whose body specifies the natural number $2$ and whose head denies
  that $2$ is a natural number. Because $P$ implies that $2$ is a natural
  number, this head is in $R$ rewritten to the empty head, enforced by
  disallowing the predicate for natural numbers in~$R$.

  \smallskip
  \subexamp{examp-thm-def-chain} %
  $\begin{array}[t]{ll}
     P = & \fc(X,Y,Z) \revimp \fr(X,Y), \fr(Y,Z). \\
         & \revimp \fc(X,Y,Z), \pnot \fr(X,Y). \\
         & \revimp \fc(X,Y,Z), \pnot \fr(Y,Z). \\[1ex]
     V = & \{\fr\} \\
  \end{array}$
  \hfill
  $\begin{array}[t]{ll}
     Q = & \fr(X,Y) \por \pnot \fr(X,Y). \\
         & \revimp \fc(X,Y,Z), \pnot \fr(X,Z). \\[1ex]
     R = & \fr(X,Y) \por \pnot \fr(X,Y). \\
         & \revimp \fr(X,Y), \fr(Y,Z), \pnot \fr(X,Z). \\
  \end{array}$

  \noindent Program $Q$ describes a transitive relation $\fr$ using the helper
  predicate $\fc$ to identify chains where transitivity
  needs to be checked. In $R$ the use of $\fc$ is not allowed, program $P$
  gives the definition of $\fc$. Similar to \cref{examp-thm-def-map}
  this realizes an unfolding of $\fc$.
\end{examp}

Definability according to \cref{thm-def} inherits potential
\emph{decidability} from the first-order entailment problem that characterizes
it. If, e.g., in the involved programs only constants occur as function
symbols, the characterizing entailment can be expressed as
validity in the decidable Bernays-Schönfinkel class.

\subsection{Constraining Positions of Predicates within Rules}
\label{sec-position}

The sensitivity of LP-interpolation to polarity inherited from Craig-Lyndon
interpolation and the program encoding with superscripted predicates offers a
more fine-grained control of the vocabulary of definitions than \cref{thm-def}
by considering also positions of predicates in rules. The following corollary
shows this.\hspace*{-1em}
\begin{coro}[Position-Constrained Effective
    Projective Definability of Logic Programs]
  \label{coro-pos}
  Let $P$ and $Q$ be logic programs and let $\VP,\VPP,\VN \subseteq \pred(P)
  \cup \pred(Q)$ be three sets of predicates. Call a logic program $R$ \defname{in
    scope of} $\la \VP,\VPP,\VN \ra$ if predicates $p$ occur in $R$ only as specified in the following table.

  \begin{center}
  \begin{tabular}{l@{\hspace{1em}}l}
  $p$ is allowed in &  only if $p$ is in\\\midrule
  Positive heads & $\VP$\\
  Negative bodies & $\VP \cup \VPP$\\
  Negative heads & $\VN$\\
  Positive bodies & $\VN$\\
  \end{tabular}\par
  \end{center}
  
  \noindent
  The existence of a logic program~$R$ in scope of $\la \VP,\VPP,\VN \ra$ with
  $\fun(R) \subseteq \fun(P) \cup \fun(Q)$ such that $P \cup R$ and $P \cup Q$
  are strongly equivalent is expressible as entailment between two first-order
  formulas. Moreover, if existence holds, such a program~$R$ can be
  effectively constructed, via a universal Craig-Lyndon interpolant of the
  left and the right side of the entailment.
\end{coro}
\begin{proof}[Sketch]
Like \cref{thm-def} but with applying the renaming of disallowed predicates
not already at the program level but in the first-order encoding with
considering polarity. Let $V^{\pm}$ be the set of polarity/predicate pairs
defined as $V^{\pm} \eqdef \{+p^0 \mid p \in \VP\} \cup \{+p^1 \mid p \in \VP
\cup \VPP\} \cup \{-p^0 \mid p \in \VN\} \cup \{-p^1 \mid p \in \VN\}$. The
corresponding entailment underlying definability and LP-interpolation is then
$\SP_{F} \land \SP_{Q} \land \pfol(P) \land \pfol(Q) \entails \lnot \SP_{P}'
\lor \lnot \SP_{Q}' \lor \lnot \pfol(P)' \lor \pfol(Q)' \lor \lnot
\mathit{Aux}$, where the primed variations of formulas are obtained by
replacing each predicate $p$ that does not appear in $V^{\pm}$ or appears in
$V^{\pm}$ with only a single polarity by a dedicated fresh predicate $p'$.
(Note that negation and priming commute.) With $W \eqdef
\predpol(\lnot \SP_{P} \lor \lnot \SP_{Q} \lor \lnot \pfol(P) \lor \pfol(Q))$
we define $\mathit{Aux}$ as
\[\bigwedge_{+p \in V^{\pm}, -p \notin V^{\pm}, +p \in W} \forall \xs\, (p(\xs) \imp
p'(\xs))\; \land \\
\bigwedge_{-p \in V^{\pm}, +p \notin V^{\pm}, -p \in W} \forall \xs\, (p'(\xs) \imp
p(\xs)),\]
where $\xs$ matches the arity of the respective predicates $p$.
\qed
\end{proof}
In general, for first-order formulas $F,G$ the second-order entailment $F
\entails \forall p\, G$, where $p$ is a predicate, holds if and only if the
first-order entailment $F \entails G'$ holds, where $G'$ is $G$ with~$p$
replaced by a fresh predicate $p'$. This explains the construction of the right side of the entailment
in the proof of \cref{coro-pos} as an encoding of
quantification upon (or ``forgetting about'') a predicate only in a
\emph{single polarity}. With predicate quantification this can be expressed,
e.g., for positive polarity, as $\exists {+p}\, G \eqdef \exists p'\, (G'
\land \forall \xs\, (p(\xs) \imp p'(\xs)))$, where $p'$ is a fresh predicate
and $G'$ is $G$ with $p$ replaced by $p'$.
We illustrate the specification of $\mathit{Aux}$ with an example.

\begin{examp}
  Assume that $+p \in V^{\pm}$ and $-p \not\in V^{\pm}$. Let $G$ stand for
  $\SP_{P} \land \SP_{Q} \land \pfol(P) \land \lnot \pfol(Q)$ and let $G'$
  stand for $G$ with all occurrences of $p$ replaced by $p'$. For now we
  ignore the restrictions of $\mathit{Aux}$ by $W$ as they only have a
  heuristic purpose.
  The following statements, which are all equivalent to each other, illustrate
  the specification of $\mathit{Aux}$ in a step-by-step fashion. We start with
  expressing the required ``forgetting''. Since it appears in a negation, we
  have to forget here the ``allowed''~$+p$. (1)~$F \entails \lnot \exists
  {+p}\, G$. (2)~$F \entails \lnot \exists p'\, (G' \land \forall \xs\,
  (p(\xs) \imp p'(\xs)))$. (3)~$F \entails \forall p'\, (\lnot G' \lor \lnot
  \forall \xs\, (p(\xs) \imp p'(\xs)))$. (4)~$F \entails \lnot G' \lor \lnot
  \forall \xs\, (p(\xs) \imp p'(\xs))$. (5)~$F \entails \lnot G' \lor \lnot
  \mathit{Aux}$.

  Observing the restrictions by membership in $W$ in the definition of
  $\mathit{Aux}$ can result in a smaller formula~$\mathit{Aux}$. Continuing
  the example, this can be illustrated as follows. Assume $+p \in V^{\pm}$ and
  $-p \not\in V^{\pm}$ as before and in addition $+p \notin W$. Since $+p
  \notin W$ it follows that $-p \notin \predpol(G)$. So ``forgetting'' about
  $+p$ in $G$ is then just $\exists p'\, G'$ and the $\mathit{Aux}$ component
  for $p$ is not needed. (Also a further simplification, outlined in the
  remark below, is possible in this case.)
\end{examp}

\begin{remark}
\label{remark-priming}
The view of ``priming'' as predicate quantification justifies a heuristically
useful simplification of interpolation inputs for definability: If a
predicate~$p$ occurs in a formula~$F$ only with positive (negative) polarity,
then $\exists p\, F$ is equivalent to $F$ with all atoms of the form $p(\tts)$
replaced by $\true$ ($\false$). Hence, if a predicate to be primed occurs only
in a single polarity, we can replace all atoms with it by a truth value
constant.
\end{remark}

We now turn to application possibilities of Corollary~\ref{coro-pos}. While it
gives some control over the position of predicates, it does not allow to
discriminate between allowing a predicate in negative heads and positive
bodies. Predicates allowed in positive heads are also allowed in negative
bodies. We give some examples.

\begin{examp}
  \label{examp-coro-pos}
  The following examples show for given programs $P,Q$ and sets $\VP, \VPP,
  \VN$ of predicates a possible value of $R$ according to
  \cref{coro-pos}. In all the examples it is essential that a
  predicate is disallowed in $R$ only in a single polarity. If it would not be
  allowed at all there would be no solution~$R$.

  \smallskip
  \subexamp{examp-coro-pos-pos}
  $\begin{array}[t]{ll}
    P = & \fp \revimp \fq.
  \end{array}$
  \hfill
  $\begin{array}[t]{ll}
    Q = & \fr \revimp \fp.\\
    & \fr \revimp \fq.\\
    & \fq \revimp \fs.
  \end{array}$
  \hfill
  $\begin{array}[t]{ll}
    \VP = & \{\fp,\fq,\fr,\fs\}\\
    \VPP = & \{\}\\
    \VN = & \{\fp,\fr,\fs\}\\
  \end{array}$
  \hfill
  $\begin{array}[t]{ll}
    R = & \fr \revimp \fp.\\
        & \fq \revimp \fs.
  \end{array}$

  \noindent
  Here $\fq$ is allowed in $R$ in positive heads (and negative bodies) but not
  in positive bodies (and negative heads). Parentheses indicate constraints
  that apply but are not relevant for the example.

  \smallskip
  \subexamp{examp-coro-pos-neg}
  $\begin{array}[t]{ll}
    P = & \fp \revimp \fq.
  \end{array}$
  \hfill
  $\begin{array}[t]{ll}
    Q = & \revimp \fq, \pnot \fp.\\
    & \fr \revimp \fq.\\
    & \fs \revimp \fp.
  \end{array}$
  \hfill
  $\begin{array}[t]{ll}
    \VP = & \{\fq,\fr,\fs\}\\
    \VPP = & \{\}\\
    \VN = & \{\fp,\fq,\fr,\fs\}\\
  \end{array}$
  \hfill
  $\begin{array}[t]{ll}
    R = & \fr \revimp \fq.\\
        & \fs \revimp \fp.
  \end{array}$

  \noindent
  Here $\fp$ is allowed in $R$ in positive bodies (and negative heads) but not
  in negative bodies (and positive heads).

  \smallskip
  \subexamp{examp-coro-pos-phnb}
  $\begin{array}[t]{ll}
    P = & \fp \revimp \fq.\\
        & \fr \revimp \fp.
  \end{array}$
  \hfill
  $\begin{array}[t]{ll}
    Q = & \fs \revimp \pnot \fr.\\
    & \fr \revimp \fq.\\
  \end{array}$
  \hfill
  $\begin{array}[t]{ll}
    \VP = & \{\fs\}\\
    \VPP = & \{\fr\}\\
    \VN = & \{\fp,\fq,\fr,\fs\}\\
  \end{array}$
  \hfill
  $\begin{array}[t]{ll}
    R = & \fs \revimp \pnot \fr.
  \end{array}$

  \noindent
  Here $\fr$ is allowed in $R$ in negative bodies and but not in positive
  heads.

\end{examp}

\section{Prototypical Implementation}
\label{sec-implementation}

We implemented the synthesis according to \cref{thm-def} and
\cref{coro-pos} prototypically with the \PIE (\name{Proving,
  Interpolating, Eliminating}) environment~\cite{cw:pie:2016,cw:pie:2020},
which is embedded in SWI-Prolog~\cite{swiprolog}. The implementation and all
requirements are free software, see
\url{http://cs.christophwernhard.com/pie/asp}.

For Craig-Lyndon interpolation there are several options, yielding different
solutions for some problems. Proving can be performed by \CMProver, a clausal
tableaux/connection prover connection~%
\cite{letz:habil,bibel:atp,bibel:otten:2020} included in \PIE, similar to
\PTTP~\cite{pttp}, \SETHEO~\cite{setheo:92} and \leanCoP~\cite{leancop}, or by
\ProverN~\cite{prover9}. Interpolant extraction is performed on clausal
tableaux following~\cite{cw:interpolation:2021}. Resolution proofs by \ProverN
are first translated to tableaux with the \name{hyper} property, which
allows to pass range-restriction and the Horn property from inputs to outputs
of interpolation~\cite{cw:range:2023}. Optionally also proofs by \CMProver can
be transformed before interpolant extraction to ensure the hyper property.
With \CMProver it is possible to enumerate alternate interpolants extracted
from alternate proofs. More powerful provers such as \EProver~\cite{eprover}
and \Vampire~\cite{vampire} unfortunately do not emit gap-free proofs that
would be suited for extracting interpolants.

The organization of the implementation closely follows the abstract exposition
in \cref{sec-craig-beth}, with Prolog predicates corresponding to theorems.
For convenience in some applications, the predicate that realizes
\cref{thm-def} and \cref{coro-pos} permits to specify the vocabulary also
complementary, by listing predicates not allowed in the result.
In general, if outputs are largely \emph{semantically} characterized,
\emph{simplifications} play a key role. Solutions with redundancies should be
avoided, even if they are correct. This concerns all stages of our workflow:
preparation of the interpolation inputs, choice or transformation of the proof
used for interpolant extraction, interpolant extraction, the interpolant
itself, and the first-order representation of the output program, where strong
equivalence must be preserved, possibly modulo a background program. Although
our system offers various simplifications at these stages, this seems an area
for improvement with large impact for practice. Some particular issues only
show up with experiments. For example, for both \CMProver and \ProverN a
preprocessing of the right sides of the interpolation entailments to reduce
the number of distinct variables that are Skolemized by the systems was
necessary, even for relatively small inputs.

The application of first-order provers to interpolation for reformulation
tasks is a rather unknown territory. Experiments with limited success are
described in~\cite{benedikt:2017}. Our prototypical implementation covers the
full range from the application task, synthesis of an answer set program for
two given programs and a given vocabulary, to the actual construction of a
result program via Craig-Lyndon interpolation by a first-order prover. At
least for small inputs such as the examples in the paper it successfully
produces results. We expect that with larger inputs from applications it at
least helps to identify and narrow down the actual issues that arise for
practical interpolation with current first-order proving technology. This is
facilitated by the embedding into \PIE, which allows easy interfacing to
community standards, e.g., by exporting proving problems underlying
interpolation in the \name{TPTP}~\cite{tptp} format.

\section{Conclusion}
\label{sec-conclusion}

We presented an effective
variation of projective Beth definability based on Craig
interpolation for answer set programs with respect to strong equivalence under
the stable model semantics.
Interpolation theorems for logic programs under stable models semantics were
shown before in~\cite{amir:2002}, where, however, programs are only on the
left side of the entailment underlying interpolation, and the right side as
well as the interpolant are just first-order formulas. Craig interpolation and
Beth definability for answer set programs was considered in
\cite{gpv:interpolable:2011,pearce:valverde:synonymous:2012}, but with just
existence results for equilibrium logic, which transfer to answer set
semantics. The transfer of Craig interpolation and Beth definability from
monotonic logics to
default logics is investigated
in~\cite{cassano:etal:2019}, however, applicability to the stable model
semantics and a relationship to strong equivalence are not discussed.

In~\cite{toman:weddell:horn:2022} ontology-mediated querying over a knowledge
base for specific description logics is considered, based on Beth definability
and Craig interpolation. Interpolation is applied there to the Clark's
completion~\cite{clark:1978} of a Datalog program. Although completion
semantics is a precursor of the stable model semantics, both agreeing on a
subclass of programs, completion seems applied in~%
\cite{toman:weddell:horn:2022} actually on program fragments, or on the
``schema level'', as in our \cref{examp-thm-def-map,examp-thm-def-map-fold,examp-thm-def-chain}.
A systematic
investigation of these forms of completion is on our agenda.
\name{Forgetting}~\cite{delgrande:17,GoncalvesEtAl2023} or, closely related,
\name{uniform interpolation} and \name{second-order quantifier elimination}~%
\cite{soqe:book}, may be seen as generalizing Craig interpolation: an
expression is sought that captures exactly what a given expression says about
a restricted vocabulary. A \emph{Craig} interpolant is moreover related to
a second given expression entailed by the first, allowing to extract it from a
proof of the entailment.

We plan to extend our approach to classes of programs with practically
important language extensions. Arithmetic expressions and comparisons in rule
bodies are permitted in the language mini-\textsc{gringo}, used already in
recent works on verifying strong equivalence~%
\cite{flls:verifying:tight:2020,Lifschitz2022,FandinnoEtAl2023}. We considered
strong equivalence relative to a context \emph{program}~$P$. These contexts
might be generalized to first-order theories that capture
theory extensions of logic programs~%
\cite{GebserEtAl2016a,JanhunenEtAl2017,CabalarEtAl2019b}.

So far, our approach characterizes result programs syntactically by
restricting allowed predicates and, to some degree, also their positions in
rules. Can this be generalized? Restricting allowed functions, including
constants, seems not possible: If a function occurs only in the left side of
the entailment underlying Craig interpolation, the interpolant may have
existentially quantified variables, making the conversion to a logic program
impossible. From the interpolation side it is known that the Horn property and
variations of range-restriction can be preserved~\cite{cw:range:2023}. It
remains to be investigated, how this transfers to synthesizing logic programs,
where in particular restrictions of the rule form and \name{safety}~%
\cite{LeeEtAl2008,CabalarEtAl2009a}, an important property related to range
restriction, are of interest.

In addition to verifying strong equivalence, recent work addresses verifying
further properties, e.g., \emph{uniform} equivalence (equivalence under inputs
expressed as ground facts)~%
\cite{LifschitzEtAl2018,flls:verifying:tight:2020,FandinnoEtAl2023}. The
approach is to use completion~\cite{clark:1978,lifschitz:book} to express the
verification problem in classical first-order logic. It is restricted to
so-called \emph{locally tight} logic programs~\cite{FandinnoEtAl2024}. Also
forms of equivalence that abstract from ``hidden'' predicates are mostly
considered for such restricted program classes, as relative
equivalence~\cite{lin:equivalence:2002}, projected answer
sets~\cite{eiter:solution:correspondences:2005}, or external
behavior~\cite{FandinnoEtAl2023}. It remains future work to consider
definability with uniform equivalence and hidden predicates, possibly using
completion for translating logic programs to formulas (instead of $\pfol$),
although it applies only to restricted classes of programs.

Independently from the application to program synthesis, our characterization
of \name{encodes a program} and our procedure to extract a program from a
formula suggest a novel practical method for transforming logic programs while
preserving strong equivalence. The idea is as follows, where $P$ is the given
program: \emph{First-order transformations} are applied to $F \eqdef \pfol(P)$
to obtain a first-order formula~$F'$ such that $\SP_{F} \land F' \equiv \SP_{F} \land
F$. For transformations that result in a universal formula, $F'$
encodes a logic program, as argued in \cref{remark-simp}. Applying the
extraction procedure to $F'$ then results in a program $P'$ that is strongly
equivalent to~$P$. This makes the wide range of known first-order
simplifications and formula transformations applicable and provides a firm
foundation for soundness of special transformations. We expect that this
approach supplements known dedicated simplifications that preserve strong
equivalence~\cite{eiter:solution:correspondences:2005,BrassDix1999}.

With its background in artificial intelligence research, answer set
programming is a declarative approach to problem solving, where specifications
are processed by automated systems. It is suitable for meta-level reasoning to
verify properties of specifications and to synthesize new specifications. On
the basis of a technique to verify an equivalence property of answer set
programs we developed a synthesis technique. Our tools were Craig
interpolation and Beth definability, fundamental insights about first-order
logic that relate given formulas to further formulas characterized in certain
ways. Practically realized with automated first-order provers, Craig
interpolation and Beth definability become tools to synthesize formulas, and,
as shown here, also answer set programs.
\subsubsection{Acknowledgments.}
 The authors thank anonymous reviewers for helpful suggestions to improve the
 presentation.

\bibliographystyle{splncs04}
\bibliography{bibexport} 

\closeout\plabelsfile
\closeout\plabelslogfile

\end{document}